\documentclass[letterpaper, 10 pt, conference]{ieeeconf}  

\IEEEoverridecommandlockouts                              

\usepackage{graphicx} 
\usepackage{covington}

\usepackage[linesnumbered,ruled]{algorithm2e}
\SetKwInput{KwInput}{Input}                
\SetKwInput{KwInit}{Initialize}                
\SetKwInput{KwSet}{Set}                
  \SetKwFunction{FMain}{Main}
  \SetKwFunction{FBreak}{Break$\_$Coalitions}
  \SetKwFunction{Prop}{Generate$\_$Proposal}
  \SetKwFunction{Free}{Update$\_$Free}
\SetKwBlock{Loop}{loop}{end}
\let\labelindent\relax
\usepackage{enumitem}
\usepackage{amsmath} 

\usepackage{amsthm} 
\usepackage{float}
\usepackage{amssymb}  

\newtheorem{definition}{Definition}
\newtheorem{proposition}{Proposition}
\newtheorem{assumption}{Assumption}

\newtheorem{thm}{Theorem}

\usepackage{nomencl}

\title{\LARGE \bf Distributed Learning Dynamics Converging to the Core of $B$-Matchings
}

\author{Aya Hamed$^{1}$ and Jeff S. Shamma$^{1}$
\thanks{*This work was supported by the University of Illinois Urbana-Champaign}
\thanks{$^{1}$Aya Hamed and Jeff S. Shamma are with the Department of Industrial and Enterprise Systems Engineering, the Grainger College of Engineering, University of Illinois at Urbana-Champaign, 117 Transportation Building, MC-238, 104 S. Mathews Ave. Urbana, IL 61801-3080, USA {\tt\small ayah2@illinois.edu, jshamma@illinois.edu}}%
}

\begin{document}

\maketitle
\thispagestyle{empty}
\pagestyle{empty}

\begin{abstract}

 $B$-Matching is a special case of matching problems where nodes can join multiple matchings with the degree of each node constrained by an upper bound, the node's $B$-value. The core solution of a bipartite $B$-matching is both a matching between the nodes respecting the upper bound constraint and an allocation of the weights of the edges among the nodes such that no group of nodes can deviate and collectively gain higher allocation. We present two learning dynamics that converge to the core of the bipartite $B$-matching problems. The first dynamics are centralized dynamics in the nature of the Hungarian method, which converge to the core in a polynomial time. The second dynamics are distributed dynamics, which converge to the core with probability one. For the distributed dynamics, a node maintains only a state consisting of (i) the aspiration levels for all of its possible matches and (ii) the matches, if any, to which it belongs. The node does not keep track of its history nor is it aware of the environment state. In each stage, a randomly activated node proposes to form a new match and changes its aspiration based on the success or failure of its proposal. At this stage, the proposing node inquires about the aspiration of the node it wants to match with to calculate the feasibility of the match. The environment matching structure changes whenever a proposal succeeds. A state is absorbing for the distributed dynamics if and only if it is in the core of the $B$-matching. 
\end{abstract}

\section{INTRODUCTION}

Matchings are used widely in multi-agent systems to model setups with pairwise interactions. Edges between agents can describe communication, contracts, or assignments. Matchings have been studied in various forms, connection-wise: one-to-one, many-to-one, and many-to-many, and graph-wise: general and bipartite. Constraints and parameters can also be added to limit the number of connections and the weight/cost of these connections. The different forms of matchings allow us to reduce many applications to matching problems. For example, multi-agent task assignment problems can usually be modeled as one form of matching according to the objectives of the application [1]-[3]. Matching applications also include the widely known kidney exchange problem, the medical residents matching program in the United States, and ride-sharing applications [4]-[5]. 

Various real-life applications can be modeled as many-to-many bipartite matching problems, for instance, online ad placements [6]-[7] and firms and consultants matching [8]. In addition, the authors in [9] modeled data markets where data can be replicated with zero cost as a many-to-many bipartite matching problem. Furthermore, in the networks field, many problems are related to matching through duality, reduction, or extension, such as edge covering, stable set polyhedron, flow, transportation, and transshipment problems [10]. 

The author in [11] presented some key considerations for multi-agent systems. Among these considerations are the autonomy of the multi-agent system, the complexity of the decision-making and learning, the communication between agents, the distribution of the agents, and the system's security and privacy requirements. The distributed dynamics presented in this paper address a number of these considerations since the dynamics are decentralized, require only simple computations from the agents per iteration, and reveal only minimal information about other agents' states. The setup we are considering is the bipartite $B$-matching setup, which allows many-to-many matchings with an upper bound on the number of matches per node and at most one match between each pair of nodes. In our dynamics, an agent is randomly activated which proposes to match with another agent if both of their aspirations can be fulfilled by the value of the match. The proposing agent increases or decreases its aspiration according to the success or failure of the proposal and its current matching state. An agent cannot exceed its upper bound of matches, hence, if it is already at its capacity for matches, it breaks off an old match to join a new one. 

The presented dynamics converge to the core allocation of the $B$-matching problems. The core allocation is an essential solution concept in cooperative games that guarantees that no group of agents can gain more by deviating and forming their own coalition in the game. The core allocation is proven to exist for bipartite $B$-matchings in [12]. We prove that through our distributed dynamics and only pairwise interactions, the dynamics reach a core allocation with probability one.

Our distributed dynamics, the $B$-Matching Proposals dynamics, extend and build upon the dynamics in [13]-[15]. The dynamics in [13] and [14] address the one-to-one bipartite matching setup and the many-to-one bipartite matching setup, respectively. The dynamics in [15] apply to general coalitional games with transferable utility, an example of which is the matching problems on a general graph. Similar work has been carried out in [16] for assignment games and in [17] for superadditive coalitional games. For $B$-matchings, a core allocation can be reached using the dynamics in [15] if it is formulated as a coalitional game, yet it requires the calculations of the maximum $B$-matching value (defined later) for each coalition of agents. Our current proposed dynamics bypass these calculations and use only pairwise interactions to reach the core of the $B$-matching.

The rest of this paper is organized as follows. Section II introduces some preliminaries for our setup. Section III defines the Paths Transfers algorithm and exhibits its convergence proof. Section IV introduces the $B$-Matching Proposals algorithm. Section V discusses the $B$-Matching Proposals dynamics convergence. In Section VI, we illustrate the $B$-Matching Proposals dynamics on a multi-agent task assignment setup. Finally, Section VII concludes the paper.

\section{PRELIMINARIES}

\subsection{Definitions}
Our setup in this paper is the weighted simple $B$-matching presented in [10]. Henceforth, we will refer to it as the nodes $B$-matching. The following nomenclature will be used throughout the paper.

\begin{small}
\begin{align*}
    G&=U\cup V &\; &\text{The grand union of all nodes} \\
    E&=U\times V &\; &\text{Edges between nodes} \\
[B(g)]&=\{1,2,...,B(g)\}&\;  &\text{The $B$-copies identifiers set}\\
g_k & &\quad &\text{The $k$-th copy of $g$ node}\\
\mathcal{U}&=\cup_{u\in U}\{u_i:i\in [B(u)]\}&\; &\text{copies of } U \text{ nodes}\\
\mathcal{V}&=\cup_{v\in V}\{v_j:j\in [B(v)]\}&\; &\text{copies of } V \text{ nodes}\\
\mathcal{G}&=\mathcal{U}\cup \mathcal{V} &\;  &\text{The grand union of all copies}\\
\mathcal{E}&=\mathcal{U}\times \mathcal{V} &\; &\text{Edges between copies} \\
\mathbf{x}&=\{x_g:g\in G\} &\; &\text{Nodes allocation}\\ 
\mathbf{a}_g&=\{a_{g,1},a_{g,2},...,a_{g,B(g)}\} &\; &\text{Copies allocations of node }g\\ 
\end{align*} 

\end{small}

\begin{definition}
    Consider a complete bipartite graph $(U,V,E),$ together with a weighing function $W:E\rightarrow \mathbb{Q}^+,$ and a $B$-value function $B:G \rightarrow \mathbb{Z}^+.$ Let $M\subseteq E$ and the degree of a node within $M$ be $\delta^M: G \rightarrow \mathbb{Z}^+.$ If $\delta^M(g)\leq B(g)$ for all $g\in G.$ Then $M$ is a \textbf{$B$-matching}.
\end{definition}

\begin{definition}
    Set the\textbf{ maximum $B$-matching value} of a graph $(U,V,E)$ and the associated $W$ and $B$ functions to be $\max_{M \in\mathfrak{B}}  \sum_{(u,v)\in M}W(u,v),$ where $\mathfrak{B}$ is the set of $B$-matchings over the $(U,V,E)$ graph. The maximum $B$-matching value over a coalition $S_U \cup S_V$ where $S_U\subseteq U$ and $S_V \subseteq V$ is the maximum $B$-matching value over the subgraph $(S_U,S_V,S_U\times S_V)$ with the associated $W$ and $B$ functions.
\end{definition}
 The following definition follows the definition of the core of $B$-matching in [12].  

\begin{definition}
    An allocation $\mathbf{x}$ is a \textbf{core allocation} of the $B$-matching problem if both of the following conditions are satisfied. (i) For any coalition $S_U\cup S_V,$ where $S_U\subseteq U$ and $S_V\subseteq V,$ the sum of the allocation of the nodes in the coalition is greater than or equal to the maximum $B$-matching value over the coalition. (ii) The total allocation of all nodes in $G$ is equal to the maximum $B$-matching value over the $G$ coalition, i.e. $\sum_{g\in G}x_g=\max_{M \in\mathfrak{B}} \sum_{(u,v)\in M}W(u,v)$ where $\mathfrak{B}$ is the set of $B$-matchings over the complete graph. 

    We associate with the core allocation $\mathbf{x},$ a $B$-matching $M\subseteq E$ that achieves the maximum $B$-matching value and denote the pair $(\mathbf{x},M)$ as a \textbf{nodes-core solution}.
\end{definition}

As an illustration, consider a multi-agent task assignment setting, where each agent is able to perform multiple tasks and each task can be performed by several agents. The profit that is realized from an agent performing a task depends on both the agent and the nature of the task. The previous definitions can be interpreted as follows. The two node groups $U$ and $V$ are the group of agents and the group of tasks, respectively. A connected edge $(u,v)$ is an agreement that agent $u$ performs task $v.$ The profit gained as agent $u$ performs task $v$ is the weight of the edge joining them, $W(u,v).$ The profit of matched edges can then be divided among the agents and the tasks. The $B$-value for a task is an upper bound on the maximum number of agents that can perform the task and the $B$-value for an agent is the maximum number of tasks it can perform. The allocation $\mathbf{x}$ dictates how much gain each agent and each task get from the profits collected from all matched edges.

In this paper, we are considering further a different perspective for the $B$-matching. We represent each node $g$ as $B(g)$ copies and specify the copies of the nodes that are connected. We define the $B$-matching and the core from this perspective as follows.

\begin{definition}
    Consider a complete bipartite graph $(U,V,E),$ together with a weighing function $W:E\rightarrow \mathbb{Q}^+,$ and a $B$-value $B:G \rightarrow \mathbb{Z}^+.$ Consider $\mathcal{U},\mathcal{V},\mathcal{G},$ and $\mathcal{E}$ as defined earlier. Let $\mathcal{M}\subseteq \mathcal{E},$ the degree of a node $\delta^{\mathcal{M}}(g_k)\leq1$ for all $ g\in \mathcal{G}, k\in[B(g)],$ and $|\{(u_i,v_j)\in \mathcal{M}: i\in [B(u)],j\in [B(v)]\}|\leq1$ for each $u\in U$ and $v\in V,$ then $\mathcal{M}$ is a \textbf{$B$-copies matching}.
\end{definition}

\begin{definition}
    A \textbf{reduction function} $r$ maps a pair of copies allocations and $B$-copies matching to a pair of nodes allocation and nodes $B$-matching as follows. If $(\mathbf{x},M)=r(\{\mathbf{a}_g\}_{g\in G},\mathcal{M}),$ then $M=\{(u,v): \exists i\in [B(u)],\ j\in [B(v)],\text{ such that } (u_i,v_j)\in \mathcal{M}\}$ and ${x_g}=\sum_{k\in [B(g)]}a_{g,k}.$ Note that $M$ is still a $B$-matching as previously defined because for each node $g\in G,\ \delta^M(g)\leq B(g).$  
\end{definition}

From this paper's perspective, each node $g$ is expanded as $B(g)$ distinct copies, each treated as a node in an expanded graph and gets its own gain from the profit of the matches; a node's gain then is the sum of the gains collected by all of its copies. 

\begin{definition}
    A \textbf{copies-core solution} is a pair of copies allocations and $B$-copies matching $(\{\mathbf{a}_g\}_{g\in G},\mathcal{M})$  that satisfies the following three conditions.
     \begin{enumerate}[label=\arabic*.]
        \item \textbf{Edge saturation:} \[\forall (u_i,v_j)\in \mathcal{M},\ a_{u,i}+a_{v,j}=W(u,v)\]
        \item \textbf{Stability against pairwise deviation:}\\ Let $(\mathbf{x},M)=r(\{\mathbf{a}_g\}_{g\in G},\mathcal{M}),$ then \[\forall (u,v)\notin M,\ \forall i \in [B(u)],\  j \in [B(v)],\]\[ a_{u,i}+a_{v,j}\geq W(u,v)\]
        \item \textbf{Zero-gain for unmatched copies:}  \[ \forall u_i \in \mathcal{U},\text{ if } \nexists v_j\in \mathcal{V}\text{ s.t. } (u_i,v_j)\in \mathcal{M},\text{ then }a_{u,i}=0,\] 
        \[\forall v_j \in \mathcal{V},\text{ if } \nexists u_i\in \mathcal{U}\text{ s.t. } (u_i,v_j)\in \mathcal{M},\text{ then }a_{v,j}=0\]        
    \end{enumerate}
\end{definition}

\begin{proposition}
    If $(\{\mathbf{a}_g\}_{g\in G},\mathcal{M})$ is a copies-core solution then $r(\{\mathbf{a}_g\}_{g\in G},\mathcal{M})$ is a nodes-core solution.
\end{proposition}
\begin{proof} 
     \begin{enumerate}[label=\arabic*.]
       \item Consider any $S=S_U\cup S_V$ be a subset of $G,$ and $E_S=S_U\times S_V$ be the edges over $S.$ \item Let \begin{itemize}
       \item $M_S^* \subseteq E_S$ be a max nodes $B$-matching for the set $S,$ i.e. $M^*_S=$  $\arg\max_{M_S}\sum_{(u,v)\in M_S}W(u,v),$ where $M_S$ is a $B$-matching over $S,$
       \item  $v(M_S^*)$ be the value of such $B$-matching where $v(M_S^*)=\sum_{(u,v)\in M_S^*}W(u,v),$ and
       \item $(\textbf{x},M)=r(\{\mathbf{a}_g\}_{g\in G},\mathcal{M})$
   \end{itemize}
   \item  Let  $\mathcal{M}'_1=\{(u_i,v_j):(u_i,v_j)\in \mathcal{M}$ and $(u,v)\in M_S^*\}$
   \item  Construct $\mathcal{M}'_2$ in the following way. 
   \begin{itemize}
       \item Loop over $(u,v)\in M_S^*\setminus M.$ For each $(u,v)\in M_S^*\setminus M,$ pick any free $u_i$ and $v_j,$ i.e. neither matched in $\mathcal{M}_1'$ nor in $\mathcal{M}_2',$ and add $(u_i,v_j)$ to $\mathcal{M}'_2.$
   \end{itemize}
   \item 
   \begin{itemize}
       \item   For all $(u_i,v_j)\in \mathcal{M}'_1,$  $W(u,v)=a_{u,i}+a_{v,j},$ since $(u_i,v_j)\in\mathcal{M}.$ 
       \item    For all $(u_i,v_j)\in \mathcal{M}'_2,$ since $(u,v)\notin M,$ then $\forall i \in [B(u)]$ and $j \in [B(v)], W(u,v) \leq a_{u,i}+a_{v,j}.$
   \end{itemize}
   \item This implies that 
   \[\sum_{u\in S_U,\\ i \in [B(u)]}\hspace{-6px} a_{u,i}+\sum_{v\in S_V, j \in [B(v)]}\hspace{-6px} a_{v,j}\geq \hspace{-6px} \sum_{(u,v)\in M_S^*}W(u,v).\]
    \item Due to \textbf{edge saturation} and \textbf{zero-gain for unmatched copies,} the following holds. 
    \[\sum_{u\in U}\sum_{i \in [B(u)]}a_{u_i}+\sum_{v\in V}\sum_{j \in [B(v)]}a_{v_j}= \sum_{(u_i,v_j)\in \mathcal{M}}W(u,v).\]
    Since $\mathcal{M}$ is a $B$-copies matching, then,
    \[\sum_{u\in U}\sum_{i \in [B(u)]}a_{u_i}+\sum_{v\in V}\sum_{j \in [B(v)]}a_{v_j}=\sum_{(u,v)\in M}W(u,v).\]
    From point $6,$ \[\sum_{u\in U}\sum_{i \in [B(u)]}a_{u,i}+\sum_{v\in V}\sum_{j \in [B(v)]}a_{v,j}\geq \sum_{(u,v)\in M^*}W(u,v),\] where $M^*$ is a maximum $B$-matching over $G.$ \\
    Finally, since $M^*$ is a maximum $B$-matching and the total aspirations of the copies equal to a $B$-matching, then from the previous inequality and equation, we get that \[\sum_{u\in U}\sum_{i \in [B(u)]}a_{u,i}+\sum_{v\in V}\sum_{j \in [B(v)]}a_{v,j}= \sum_{(u,v)\in M^*}W(u,v),\] where $M^*$ is a maximum $B$-matching over $G.$ 
   \end{enumerate}
   From points $6$ and $7,$ and the definition of the reduction function, then \[\sum_{u\in S_U}x_{u}+\sum_{v\in S_V}x_{v}\geq \sum_{(u,v)\in M_S^*}W(u,v),\] and 
   \[\sum_{u\in U}x_{u}+\sum_{v\in V}x_{v}= \sum_{(u,v)\in M^*}W(u,v).\]
   Hence, $(\textbf{x},M)$ is a nodes-core solution.
\end{proof}

 Consider the copies-core solution of the preceding task assignment setting. The copies-core solution in this setting specifies the allocations of the profit for each copy of the nodes. This finer detailing of the copies' allocations shows the gains of each node, robot or task, from each match it is part of, thus, decoupling the individual matches' profit distribution. In a copies-core solution, the allocation specified for each copy of a node is only its profit allocation from the match that this copy is a part of.  

In the next two sections, we define a centralized polynomial time algorithm as well as distributed dynamics that converge to the copies-core solution. Since we are aiming to reach the copies-core solution, we work with the expanded graph $(\mathcal{U},\mathcal{V},\mathcal{E}),$ as defined above. The following definitions will be used in the subsequent algorithms. Furthermore, we will overload $\mathbf{a}_g$ in the dynamics to be the agents' allocation aspirations for the copies. In our dynamics, an edge is formed only if these aspirations can be fulfilled by such match. 

\begin{definition}
    The \textbf{free copies} set in a matching $\mathcal{M}$ is defined as  $\mathcal{F}^\mathcal{M}=\{u_i:\nexists v_j \text{ s.t. }(u_i,v_j)\in\mathcal{M}\}\cup \{v_j:\nexists u_i \text{ s.t. }(u_i,v_j)\in\mathcal{M}\}.$ Node $g$ free copies are $\mathcal{F}^\mathcal{M}(g)=\{g_1,g_2,...,g_{B(g)}\}\cap \mathcal{F}^\mathcal{M}.$ Additionally, denote the set of free copies with strictly positive aspirations as $\mathcal{F}_+^{\mathcal{M}}$ and $\mathcal{F}_+^{\mathcal{M}}(g)$ for any node $g\in G.$
\end{definition}
\begin{definition}
    The \textbf{partner function} $P^\mathcal{M}$ in a matching $\mathcal{M}$ is defined as $P^\mathcal{M}:\mathcal{G}\rightarrow{2^G},$ where $P^\mathcal{M}(u_i)=\{v\}$ and $P^\mathcal{M}(v_j)=\{u\}$ if $(u_i,v_j)\in\mathcal{M},$ and $P^\mathcal{M}(g_k)=\emptyset$ if $g_k\in\mathcal{F}^\mathcal{M}.$
\end{definition}
The following is an assumption that we are going to consider for both the centralized and distributed algorithms.
\begin{assumption}
For all $u\in U,\ B(u)\leq |V|$, and for all $v\in V,\ B(v)\leq |U|;$ meaning that the nodes cannot have more copies than the number of nodes in the opposite class.  
\end{assumption}

If a given $B$-matching problem does not satisfy this assumption,  the degree of the violating nodes still can not exceed the number of nodes in the opposite class in any $B$-matching. Hence, the problem remains the same if we decrease the $B$-values of these nodes just enough to satisfy the above assumption. 

The algorithms use the aforementioned functions, however, the superscript $\mathcal{M}$ is suppressed and considered as the matching at the current iteration.

\section{PATHS TRANSFERS ALGORITHM}
In this section, we define the paths transfers algorithm, which is a centralized algorithm that is in the nature of the Hungarian method [10] and influenced by BLMA's convergence proof in [13]. The algorithm converges in a polynomial time to a copies-core solution. The algorithm proceeds as follows.

\begin{enumerate}[label=\arabic*.]
    \item Initialize the matching $\mathcal{M}$ to be empty and the aspirations of all copies to be $0$. Set $\epsilon$ to be the largest discretization value for all the weights of the edges. 
    \item \textbf{Over-aspiration step:} In this stage, the algorithm loops over all nodes in one nodes' class, without loss of generality, we will set that to be the ${V}$ class. For each $v\in {V},$ set the aspiration of all copies $v_j,\ j\in [B(v)]$ to the maximum weight of the edges that has $v$ as an endpoint, i.e. set $a_{v,j}=\max_{u\in U}W(u,v)$ for all $v\in V $and $j\in [B(v)].$

    \item \textbf{Outer loop:} Loop while $\mathcal{F}_+\neq \emptyset.$ 
    \begin{itemize}
         \item  Pick any $g_k \in\mathcal{F}_+,$ without loss of generality assume $g\in {V}$ and denote this copy by $v^*_{j^*}.$ If $g\in \mathcal{U},$ the same subsequent steps should be followed exchanging the $u_i$ and $U$ instances with the $v_j$ and $V$ instances.
        \item \textbf{Inner loop:}
        \item[] Loop so long as  $|\mathcal{F}_+|$ has not decreased.
        Consider the current matching $\mathcal{M},$ let $ (\mathbf{x},M)=r(\{\mathbf{a}_g\}_{g\in G},\mathcal{M}),$  and construct the equality directed graph $(\mathcal{U},\mathcal{V},\mathcal{E}'),$  as follows: \begin{itemize}
        \item $(u_i,v_j)\in \mathcal{E}'$ if $(u_i,v_j)\in \mathcal{M}$ 
        \item $(v_{{j}},u_{{i}})\in \mathcal{E}'$ if $(u,v)\notin M,\ a_{u,{{i}}}+a_{v,{{j}}}=W(u,v).$ 
        
    \end{itemize}
     Check the following cases in order. If one occurs, follow the steps of that case.
     \begin{enumerate}[label=\roman*.]
    \item \textbf{Decreasing aspiration case:} If there is no edge in $\mathcal{E}'$ directed from $v^*_{j^*}$ to any $u_i\in\mathcal{U},$ then decrease $a_{v^*,j^*}$ by $\epsilon.$
        \begin{itemize}
            \item [] If $a_{v^*,j^*}$ reaches zero, then $|\mathcal{F}_+|$ decreased by 1, then continue with the outer loop of step $3.$ Otherwise, reconstruct $\mathcal{E}'$ and repeat step $3's$ inner loop considering the same $v^*_{j^*}.$
            \end{itemize}

        \item  \textbf{Augmenting path case:} If there exists a directed path $\mathcal{P}$ in $\mathcal{E}',$ from the considered copy $v^*_{j^*}$ to some copy $u_{i}\in \mathcal{F}^{\mathcal{M}},$ then let $A_{tmp}=\mathcal{P}\setminus \mathcal{M}$ and  $R_{tmp}=\mathcal{P}\cap \mathcal{M}$ disregarding the directions of the edges. Add the edges in $A_{tmp}$ to $\mathcal{M}$ and remove the edges in $R_{tmp}$ from $\mathcal{M}$. 

        \begin{itemize}
            \item [] This case will lead to $|\mathcal{F}_+|$ decreasing by $1$ or $2$ depending on whether $a_{u,i}$ is greater than or equal to zero.  
        \end{itemize}
        \item \textbf{Copies exchange case:} If there exists a directed path $\mathcal{P}$ in $\mathcal{E}',$ from the considered copy $v^*_{j^*}$ to some copy $v^{**}_{j^{**}}$ with $a_{v^{**},j^{**}}=0,$ then let $A_{tmp}=\mathcal{P}\setminus \mathcal{M}$ and  $R_{tmp}=\mathcal{P}\cap \mathcal{M}$ disregarding the directions of the edges. Add the edges in $A_{tmp}$ to $\mathcal{M}$ and remove the edges in $R_{tmp}$ from $\mathcal{M}.$ 

         \begin{itemize}
            \item [] This case will lead to $|\mathcal{F}_+|$ decreasing by 1 as a copy with non-zero aspiration will be matched instead of a copy with zero aspiration.  
        \end{itemize}

        \item \textbf{Aspiration transfer case:} If none of the above cases occur, then start with $v^*_{j^*},$ and construct the following sets. Let $\mathcal{U}^{(1)}=\{u_i\in \mathcal{U}: (v^*_{j^*},u_i)\in \mathcal{E}'\}$ and  $\mathcal{V}^{(1)}=\{v_j\in \mathcal{V}: (u_i,v_j)\in \mathcal{E}', u_i\in \mathcal{U}^{(1)}\}.$ If $\mathcal{U}^{(1)}$ and $\mathcal{V}^{(1)}$ are non-empty, then let $\mathcal{U}^{(2)}=\{u_i\in \mathcal{U}: (v_j,u_i)\in \mathcal{E}',v_j\in \mathcal{V}^{(1)}\}\setminus \mathcal{U}^{(1)}$ and  $\mathcal{V}^{(2)}=\{v_j\in \mathcal{V}: (u_i,v_j)\in \mathcal{E}', u_i\in \mathcal{U}^{(2)}\}.$ Keep on constructing $\mathcal{U}^{(n)}$ and  $\mathcal{V}^{(n)},$ as long as $\mathcal{U}^{(n-1)}$ is non-empty, as follows. Let $\mathcal{U}^{(n)}=\{u_i\in \mathcal{U}: (v_j,u_i)\in \mathcal{E}',v_j\in \mathcal{V}^{(n-1)}\}\setminus \cup^{m=n-1}_{m=1} \mathcal{U}^{(m)}$ and  $\mathcal{V}^{(n)}=\{v_j\in \mathcal{V}: (u_i,v_j)\in \mathcal{E}', u_i\in \mathcal{U}^{(n)}\}.$ Let the final constructed sets be $\mathcal{U}^{(M)}$ and $\mathcal{V}^{(M)}.$ Decrease the aspiration of all copies in $\cup^{m=M}_{m=1} \mathcal{V}^{(m)},$ in addition to $v^*_{j^*},$ by $\epsilon,$ and increase the aspiration of all copies in $\cup^{m=M}_{m=1} \mathcal{U}^{(m)}$ by $\epsilon.$
        
        \begin{itemize}

            \item [] One implementation of the above steps will decrease the total aspirations of copies in $\mathcal{V}$ and increase the total aspirations of copies in $\mathcal{U}.$ The step must end once no more copies can be added to $\mathcal{U}^{(n)}$ or $\mathcal{V}^{(n)},$ which must happen since old copies in previous sets $\mathcal{U}^{(1)},\mathcal{V}^{(1)},\mathcal{U}^{(2)},\mathcal{V}^{(2)},..., \mathcal{U}^{(n-1)},\mathcal{V}^{(n-1)}$ cannot be added to $\mathcal{U}^{(n)},\mathcal{V}^{(n)}.$ This step may not lead to a decrease in $|\mathcal{F}_+|,$ if so, repeat step $3$'s inner loop starting from the same $v^*_{j^*}$ and reconstruct $\mathcal{E}'$.
            \end{itemize}
        \end{enumerate}

        \item [] The inner loop of step $3$ must end in a finite time. The inner loop will end if it goes through the \textbf{augmenting path} or \textbf{copies exchange} cases or if the considered copy reaches zero aspiration in the \textbf{decreasing aspiration} case. The inner loop cannot infinitely go through the \textbf{decreasing aspiration} and the \textbf{aspiration transfer} case since the total aspirations of copies in $\mathcal{V}$ strictly decreases with each implementation of these cases and it is bounded below by zero. Hence, in a finite number of steps, $|\mathcal{F}_+|$ will decrease by either the loop going through the \textbf{augmenting path} or \textbf{copies exchange} cases or the considered copy reaching zero aspiration. 
        
    \end{itemize}
    \item[] The outer loop of step $3$ ends in a finite time because $|\mathcal{F}_+|$ strictly decreases with each iteration.

\end{enumerate}
\begin{thm}
    The Paths Transfers algorithm reaches the copies-core solution in a finite number of steps.
\end{thm}
\begin{proof}
    First, the initialization and over-aspiration step guarantees that we reach a state $(\{\mathbf{a}_g\}_{g\in G},\mathcal{M})$ that satisfies both the \textbf{edge saturation} and \textbf{stability against pairwise deviation} conditions. The algorithm then aims to satisfy the \textbf{zero-gain for unmatched copies} condition while keeping the first two conditions satisfied. 
    After each iteration of the outer loop of step $3,$ both \textbf{edge saturation} and \textbf{stability against pairwise deviation} conditions remain satisfied. Consider one implementation of the outer loop of step $3$.
    \begin{itemize}

     \item For the decreasing aspiration case, no change in the matching occurs, and the considered copy decreases its aspiration since it strictly satisfies the \textbf{stability against pairwise deviation}. Since we decrease by $\epsilon$ on an $\epsilon$-discritized grid, the copy will remain satisfying the \textbf{stability against pairwise deviation}, and since it is not matched, it will not affect the \textbf{edge saturation} condition.
    \item For the augmenting path and copies exchange case, new connections are only made between copies whose total aspirations are exactly equal to the weight of the edge between them, since these edges belong to the equality graph. Hence, the \textbf{edge saturation} condition is preserved. Note that all copies on the paths in these cases have minimum aspirations compared to the aspirations of the other copies of their respective nodes. In addition, the summation of aspirations of these, now unmatched, copies equal to the weight of the edge between them. Thus, these copies along with other copies of these two nodes should still preserve the \textbf{stability against pairwise deviation} condition.
    \item For the aspiration transfer case, the matching does not change, and for each edge considered that is both in the matching and the equality graph, one endpoint decreases its aspiration by epsilon while the other endpoint increases its aspiration by epsilon, hence keeping their total aspiration equal to the weight of the edge and preserving the \textbf{edge saturation} condition. In addition, for each copy $g_k$ that decreases its aspiration, any copy of the opposite class that may break the \textbf{stability against pairwise deviation} condition, tightly satisfies the condition with $g_k$ before the decrease in its aspiration, increases its aspiration by $\epsilon,$ hence, preserving this condition as well.
    \end{itemize}
    Finally, since the outer loop of step $3$ in the algorithm strictly decreases $|\mathcal{F}_+|$ then the algorithm will eventually stop and the \textbf{zero-gain for unmatched copies} condition will be satisfied. Hence, the final state will be a copies-core solution.
\end{proof}

Note that the Paths Transfers algorithm is polynomial in the input parameters including the weight values but can be modified to be polynomial in $|\mathcal{E}|$ and $|\mathcal{G}|$ and exclude the weight values dependency. To achieve that, the \textbf{aspiration transfer case} should be modified such that the copies aspirations change by the minimum change of aspiration that leads to a new edge in the equality graph to appear or for one of the copies to hit zero aspiration. The \textbf{decreasing aspiration case} should be modified such that the decrease in aspiration is the minimum decrease that either leads the reconstructed $\mathcal{E}'$ to contain an edge from the considered copy $v^*_{j^*}$ to any other copy or leads $a_{v^*,j^*}$ to reach zero.

\section{ $B$-MATCHING PROPOSALS ALGORITHM}
In this section, we introduce the $B$-Matching Proposals dynamics. Intuitively, the algorithm follows these steps:

\begin{enumerate}[label=\arabic*.]
    \item Agents come with arbitrary non-negative initial aspirations for all their copies. The aspirations value should lie on a grid of width $\epsilon$, a discretization value for all the weights of the edges.
    \item The algorithm then iterates over the following steps:
    \begin{enumerate}[label=\roman*.]
        \item An agent is activated uniformly at random, which in turn chooses, using a uniform distribution, an agent in the opposite class of nodes to propose forming a match with. 
        \item If the current agents are matched, the proposal is skipped, otherwise, the proposing agent asks for and receives an aspiration value from the other agent. 
        \item The receiving agent sends the minimum aspiration of its unmatched copies if any exist, otherwise, it sends the minimum aspiration of the matched copies. Similarly, the proposing agent considers the minimum aspiration of its unmatched copies if any exist, otherwise, it considers the minimum aspiration of the matched copies. 
    If the total of the received aspiration in addition to the proposing agent's considered aspiration raised by $\epsilon$ is less than or equal to the proposed match weight, a match between the two copies of the considered aspirations is formed, otherwise, the proposal fails.
    \begin{enumerate}[label=\alph*.]
        \item If a match is successfully formed, the proposing agent increases its aspiration to be the difference between the weight of the edge joining the proposer and the receiver and the receiver's aspiration. If any of the copies in the new match were already matched, they dissolve their previous matches first to join the new one.
        \item If a match proposal fails, the proposing agent decreases its aspiration in one of the copies that are not matched and has an aspiration greater than zero, if any exist, by $\epsilon.$         
    \end{enumerate} 
    \end{enumerate}
\end{enumerate}
\begin{algorithm}[!htbp]
    \DontPrintSemicolon
    \renewcommand{\thealgocf}{}
    \caption{$B$-Matching Proposals Algorithm}
   \KwInput{ two sets of agents $U,V,$ edges weights, $W,$ the $B$-value, and the algorithm horizon, $H.$ }
   \KwSet{$\epsilon$ to be a discretization of the $W$ values.}
   \KwInit{the partners $P(g_k)\leftarrow\emptyset\ \ \forall g_k \in \mathcal{G};$ \linebreak the aspirations $ a_{g,k}\leftarrow a^0_{g,k},$ an arbitrary value on the $\epsilon-$grid, $\ \forall g\in {G}, k \in [B(g)]$; \linebreak the free copies $\mathcal{F}(g)=[B(g)]\ \ \forall g \in G$.}
    \For{h=1:H}{ Choose $p\in G$ uniformly at random.\\ \eIf{$p\in U$}{ Choose uniformly at random $r\in V$\\ $u=p,\ v=r$\\}{ Choose uniformly at random $r\in U$ \\ $v=p,\ u=r$} \vspace{8px}
    \If{$\nexists k \in [B(p)]$ s.t. $P(p_k)=\{r\}$}{{\eIf{$\mathcal{F}(r)\neq \emptyset$}{$j=\arg \min_{j'\in \mathcal{F}(r)}{a_{r,{j'}}}$} {$j=\arg \min_{j'\in [B(r)]}{a_{r,{j'}}}$}} ReceiverAspiration: $\underline{a}_r=a_{r,j}$\\ \vspace{8px}{\eIf{ $\mathcal{F}(p)\neq \emptyset$}{$i=\arg \min_{i'\in \mathcal{F}(p)}{a_{p,{i'}}}$} {$i=\arg \min_{i'\in [B(p)]}{a_{p,{i'}}}$}} ProposerAspiration: $\underline{a}_p=a_{p,i}$\\ \vspace{8px} \eIf{$\underline{a}_p+\underline{a}_r<W(u,v)$}{\vspace{0.8px} BreakOldConnections:\\ \ \  \textbf{if } $P(p_i)=\{q\}$ \\ \ \ \ \ \ \textbf{find}  $k:P(q_k)=\{p\}$ \textbf{then} \\ \quad \ \   $P(q_k)\leftarrow\emptyset$ and $\mathcal{F}(q)\leftarrow\mathcal{F}(q)\cup \{q_k\}$ \\ \ \ \textbf{end } \ \ \textbf{if} $P(r_j)=\{s\}$\\ \quad \ \ \textbf{find } $l: P(s_l)=\{r\}$ \textbf{ then } \\ \quad \ \ $P(s_l)\leftarrow\emptyset$ and $\mathcal{F}(s)\leftarrow\mathcal{F}(s)\cup \{s_l\}$ \\ \ \  \textbf{end } \\ MakeNewConnection:\\ \quad $P({p_i})\leftarrow \{r\},\ P({r_j})\leftarrow \{p\}, $  \\$ \quad \mathcal{F}(p)\leftarrow \mathcal{F}(p)\setminus \{p_i\},$ \\ \quad$ \mathcal{F}(r)\leftarrow \mathcal{F}(r)\setminus \{r_j\}  $\\ IncProposerAspiration:\\ $\quad a_{p,i}=W(u,v)-\underline{a}_r$}{ \If {$\mathcal{F}_{+}(p)\neq \emptyset$}{$m= \arg \min_{i'\in \mathcal{F}_{+}(p)}{a_{p,{i'}}}$\\ $a_{p,{m}}=a_{p,{m}}-\epsilon$}
     }
     } 
     }
\end{algorithm}
\section{$B$-MATCHING PROPOSALS ALGORITHM ANALYSIS}
\begin{definition}
    A state $(\{\mathbf{a}_g\}_{g\in G},\mathcal{M})$ is \textbf{feasible} if $\mathcal{M}$ is a $B$-copies matching, $a_{g,k}\geq 0$ for all $ g\in G, k\in [B(g)],$ and $a_{u,i}+a_{v,j}= W(u,v) $ for all $ (u_i,v_j)\in \mathcal{M}.$
\end{definition}
\begin{proposition}
    Starting from a feasible state and following the $B$-Matching Proposals algorithm, the state remains feasible.
\end{proposition}

\begin{proof}
    At each iteration, one proposal is made, if the proposal succeeds, then the aspiration of the copies forming the new matching exactly equals the weight of the edge between them. In addition, the aspirations of the agents remain on the $\epsilon$-grid. If the proposal fails, only agents with strictly positive aspiration decrease their aspiration by $\epsilon,$ and since the aspirations are on the $\epsilon$-grid, the aspiration will either be greater than or equal to zero. Furthermore, if there exist copies of a pair of nodes that are already matched, no further copies for the same nodes can match, and each copy is an endpoint of at most one edge in the match, hence, the $B$-copies matching and feasibility conditions are preserved.  
\end{proof}

\begin{thm}
    The $B$-Matching Proposals algorithm converges to the copies-core solution and hence, by reduction, to the nodes-core solution with probability one.
\end{thm}
\begin{proof}
        The proof follows three steps illustrated in the following subsections. First, starting from any feasible state, there exists a sequence of proposals that reaches a copies-core solution within some finite time $T$. Second, the probability of the occurrence of such sequence goes to one. Third, a state is absorbing if and only if it is a copies-core solution.
\end{proof}

\subsection*{I. There exists a sequence of proposals that reaches a copies-core solution from any feasible state}
\begin{proposition}
    Starting from any feasible state, there exists a finite sequence of proposals resulting in a state satisfying the \textbf{edge saturation} and \textbf{stability against pairwise deviation} conditions. 
\end{proposition}
\begin{proof}
    Each node can propose with a positive probability to another node of the opposite class that it is not currently matched with. Consider the sequences of proposals, where nodes (from one class) propose to form an edge only whenever there is a possible gain in the proposing agent's aspiration. Such sequences will have a finite number of proposals due to the boundedness of the number of edges, and hence, will occur with positive probability bounded away from zero within a time horizon equal to the number of edges, $|E|$. Any of these sequences of proposals will result in a state satisfying \textbf{edge saturation}, due to the states' feasibility, and \textbf{stability against pairwise deviation} conditions, since these sequences only stop when there are no more proposals that can result in aspiration gain.
\end{proof}

\begin{proposition}
    Starting from any feasible state satisfying the \textbf{edge saturation} and \textbf{stability against pairwise deviation} conditions, there exists a finite sequence of proposals resulting in a \textbf{copies-core} solution. 
\end{proposition}
\begin{proof}
    The proof follows from the fact that there exists a sequence of proposals that occurs with positive probability, within a finite time horizon that is polynomial to the input parameters. Such possible sequence corresponds to the implementation of step $3$ in the paths transfers algorithm. This sequence consists of subsequences that are analogous to performing the outer loop of step $3,$ and lead to $|\mathcal{F}_+|$ strictly decreasing after each subsequence. At the beginning of each subsequence, similar to the paths transfer algorithm, the equality graph, and the cases that the state can be in, guide the possible subsequence of proposals that leads to a decrease in $|\mathcal{F}_+|$.
    
    The first part of each subsequence is analogous to the \textbf{decreasing aspiration case} as follows. For each node with free copies that have aspirations greater than zero, at least one of these copies should be part of the equality graph. If that is not the case for one node, then the proposal sequence in this part is as follows. Without loss of generality assume such node belongs to $U,$ and denote it by $u,$ then consider $v\in V$ such that $(u,v)\notin M,$ then let $u$ propose to $v$ repeatedly until either $\exists v'\in V$ such that $(u,v')\notin M$ and $\min_{i\in \mathcal{F}_{+}(u)}{a_{u,{i}}}+\min_{j\in [B(v')]}{a_{v',{j}}}=W(u,v')$ or a copy of $u$ hits zero aspiration from the failed proposals leading to a strict decrease in $|\mathcal{F}_+|.$ Recalculate the equality graph, and repeat for any other node with the same violation, exchanging the $u$ and $U$ instances with $v$ and $V$ if such node belongs to $V$ instead of $U.$ If $|\mathcal{F}_+|$ decreases after the first part, then the subsequence ends then. Otherwise, the subsequence has a second part that corresponds to one of the following three cases.

    \begin{itemize}

        \item  If there exists a directed path on the equality graph from a free copy from one class to a copy of the same class with zero aspiration, then a set of proposals analogous to the \textbf{copies exchange case} is followed. The proposals are such each node on this path proposes to the subsequent node on the path twice, starting from the free copy tailing this path and excluding the proposal from the second to last node to the last node. The subsequence ends after this set of proposals since it leads to $|\mathcal{F}_+|$ strictly decreasing. 

        \item If such directed path does not exist but an augmenting path exists from a free copy with positive aspirations, then the case becomes analogous to the \textbf{augmenting path case}. A similar set of proposals is followed where each node on this path proposes to the subsequent node on the path twice, starting from the free copy tailing this path and including the second to last node proposing to the last node. The subsequence ends after this set of proposals since it leads to $|\mathcal{F}_+|$ strictly decreasing. 

        \item If none of these paths exist, then a set of proposals analogous to the operation in the \textbf{aspiration transfer case} is followed. The following set of proposals leads to the total aspirations of one class of nodes decreasing and the total aspirations of the other class increasing. One set of such proposals does not guarantee a decrease in $|\mathcal{F}_+|,$ yet repeating the set with the same class of nodes as the starting point would eventually either lead to $|\mathcal{F}_+|$ decreasing and the subsequence ends then or to a directed path on the updated equality graph to appear following one of the cases that were mentioned earlier. If such directed path appears, the corresponding set of proposals can be followed to decrease $|\mathcal{F}_+|$ and end the subsequence. The aforementioned set of proposals that is analogous to the operation in the {aspiration transfer case} proceeds as follows. 

\begin{itemize}
    \item  This set of proposals is guided by choosing a node with a free copy of positive aspiration from one class of nodes. Without loss of generality, consider such node a $V$ node denoted by $v^*,$ and $\mathcal{U}^{(n)}$ and $\mathcal{V}^{(n)}$ are constructed in the same manner as in the Paths Transfers algorithm.

    \item   The first proposal in this set is the proposal of $v^*$ to a node $u$ such that $(u,v^*)\notin M$ and $\min_{i\in [B(u)]}{a_{u,{i}}}+\min_{j\in \mathcal{F}_{+}(v^*)}{a_{v^*,{j}}}=W(u,v^*).$ The proposal fails because of the satisfaction of the stability against pairwise condition. Hence, the aspiration of a free copy, with minimum aspiration, of $v^*$ decreases. Now let $N_{v^*}$ be the node neighbors of $v^*$ through the equality graph, i.e. all the nodes that have copies as heads for edges that have a copy of $v^*$ as the tail. Call the following set of proposals a \textbf{U proposal round}. Let each node in $u\in N_{v^*}$ propose to $v^*$ in order. The proposals should succeed, old matches get broken, and the aspiration for a copy of each proposing node increases. After this round of proposals each $v$ that had a match broken due to these proposals tries to restore that match, denote this round by the \textbf{matches restoration round}. In this round, nodes in $V$ try to restore broken matches by proposing back to their old partners. The proposals initially fail and the total aspiration of the proposing nodes decreases. Another two proposals are made in an attempt to restore these matches, whenever possible by either the node from $U$ or $V,$ whichever node has free copies. The U proposal and match restoration rounds are repeated as needed till $\forall (u,v^*)\notin M,$ $u\in N_{v^*},$  $\min_{i\in [B(u)]}{a_{u,{i}}}+\min_{j\in \mathcal{F}_{+}(v^*)}{a_{v^*,{j}}}=W(u,v^*)$ and $|\mathcal{F}_+|$ either stay the same or decreases compared to its value at the beginning of the subsequence.  
    \item Repeat the previous U proposal and match restoration rounds but for each of $ v\in \mathcal{V}^{(n)}, \forall n>0 $ instead of $v^*.$ If the considered $v$ had a match before beginning the rounds that are pertaining to it, then after finishing these rounds, the $U$ node that was connected to $v$ proposes twice to it to restore that old match after the end of these rounds.
\end{itemize}
    \end{itemize}
  
\end{proof}

\subsection*{II Borel-Cantelli lemma}

The sequences in propositions $3$ and $4$ occur with a positive probability, bounded away from zero, within some finite time $T$ that does not depend on the state. Hence, by the Borel-Cantelli lemma, the probability of the occurrence of a sequence of proposals that lead to the copies-core solution goes to one as time goes to infinity.

\subsection*{III. A state is absorbing in the $B$-Matching Proposals dynamics if and only if it is a copies-core solution}

     \subsubsection*{\textbf{If a state is absorbing then it is a copies-core solution }}
       Assume the state $(\{\mathbf{a}_g\}_{g\in G},\mathcal{M})$ is absorbing. Let $(\mathbf{x},M)=r(\{\mathbf{a}_g\}_{g\in G},\mathcal{M}).$ For a state to be absorbing then,   

     \begin{itemize}
        \item the \textbf{edge saturation} property is satisfied since it is satisfied for all feasible states. 
        
         \item $\mathcal{F}_+=\emptyset$ and the $\textbf{zero-gain for unmatched copies}$ condition is satisfied. Otherwise, $\exists g \in G $ such that $\mathcal{F}_+(g)\neq \emptyset$ and due to Assumption $1,$ there exists a node in the opposite class that is not connected to $g$ through any copy. If $g$ proposes to match with such node, the proposal must fail since the state is absorbing and the matching structure cannot change, yet that leads to one copy from $\mathcal{F}_+(g)$ decreasing its aspiration, which changes the state, and hence provides a contradiction.  
         
         \item Given $ \mathcal{F}_+= \emptyset$, then $\forall (u,v)\notin M,\ i\in [B(u)],$ $j\in [B(v)]$  $a_{u_i}+a_{v_j}\geq W(u,v)$ and the $\textbf{stability against pairwise deviation}$ is satisfied. Otherwise, $\exists u,v,i$ and $j$ such that $a_{u_i}+a_{v_j} < W(u,v).$ Thus, if $u$ proposes to $v$ or vice versa, the proposal must succeed, consequently changing the matching and the aspirations, which establishes that the state is not absorbing and provides a contradiction.
     \end{itemize}
    Consequently, $(\{\mathbf{a}_g\}_{g\in G},\mathcal{M})$ is a copies-core solution.\\

     \subsubsection*{\textbf{If a state is a copies-core solution then it is absorbing}}
     Assume that the state $(\{\mathbf{a}_g\}_{g\in G},\mathcal{M})$ is a copies-core solution at iteration $h$. Assume the proposal happening at iteration $h+1,$ is from some node $u\in U$ to any node $v\in V.$ 

     \begin{itemize}
          \item If $u$ and $v$ were already connected, no change occurs.
         \item If $u$ and $v$ are not connected, then by the \textbf{stability against pairwise deviation} property, $\forall i\in B[(u)]$ and $j\in [B(v)],$ $a_{u_i}+a_{v_j}\geq W(u,v).$ Hence, all proposals must fail and no change can occur in the matching. In addition, from the $\textbf{zero-gain for unmatched copies}$ condition,  $\mathcal{F}_+= \emptyset$, then none of the copies can decrease its aspiration.

     \end{itemize}
    Similar justification applies if a node in $V$ proposes to any node in $U.$ Hence, $(\{\mathbf{a}_g\}_{g\in G},\mathcal{M})$ is absorbing.

\section{ILLUSTRATION} In this section, we demonstrate the distributed learning dynamics on the aforementioned multi-agent task assignment setting as follows.
\begin{itemize}
    \item We consider a set of tasks, $T,$ where \begin{itemize}
        \item each task has a value $V(t),$ such that $V:T\rightarrow \mathbb{Z}^+,$ and 
        \item each task can afford to utilize up to $B(t)\in \mathbb{Z}^+$ robots.
        
    \end{itemize}
       \item We consider a set of robots, $R,$ where \begin{itemize}
        \item each robot has accuracy $I(r),$ such that $I:R\rightarrow \mathbb{Q}^+,$ and
        \item each robot can perform up to $B(r)\in \mathbb{Z}^+$ tasks.
    \end{itemize}   

     \item[] In our setup, we are assuming a negative correlation with the sophistication of the robot $I(r)$ and how many tasks it can perform $B(r).$
   
    \item The value of a matching of task $t$ with robot $r$ is $W(t,r),$ where $W(t,r)$ is proportional to the accuracy of the robot multiplied by the value of the task.
\end{itemize}
We applied the $B$-Matching Proposals dynamics and examined the evolution of the matching and copies' aspirations.

\subsection{Sample run and average performance}
    Fig. 1 illustrates a sample run of the dynamics on the task assignment setting. The left side of Fig. 1 shows the final matching between the robot nodes, with each node $r$ expanded as $B(r)$ copies, and the task nodes, with each node $t$ expanded as $B(t)$ copies. The right side of Fig. 1 is a plot of the total feasible aspirations, aspirations of the matched copies, at each iteration. The dashed line represents the maximum $B$-matching value solved by a linear program.
 
        To observe the average performance of the dynamics on the preceding setting, we ran the dynamics for $100$ different configurations, setting $|T|=10,\ |R|=5,\ \epsilon=1,$ and randomizing the $V$ and $B$-values of the tasks. Fig. 2 shows the plot of the average of the relative total feasible aspirations, where, for each run, the relative total feasible aspiration is the total feasible aspirations divided by the optimal value for that run.
          \begin{figure}[h]
    \begin{center}
        \includegraphics[scale=0.3]{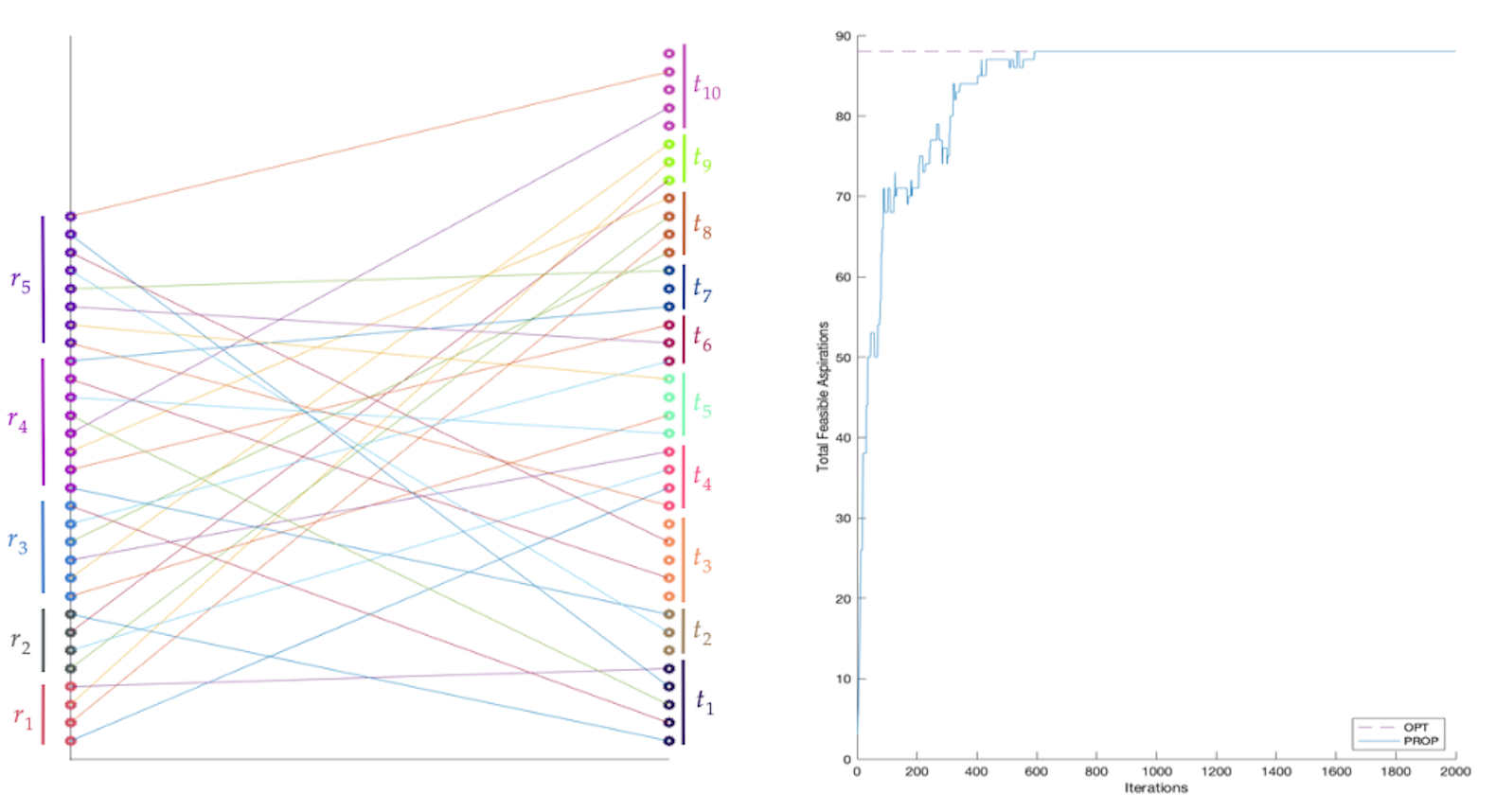}
    \caption{Sample run of the multi-robot task assignment setting. On the right, the optimal value `OPT' is plotted in purple and the total feasible aspirations of the agents, as the $B$-Matching Proposals dynamics evolve, `PROP' is plotted in blue. On the left, an illustration of the final $B$-copies matching between the robots and the tasks is shown.}
    \label{fig:enter-label}
  
    \end{center}
    \end{figure}
\begin{figure}[h]
    \begin{center}
    
    \includegraphics[scale=0.15]{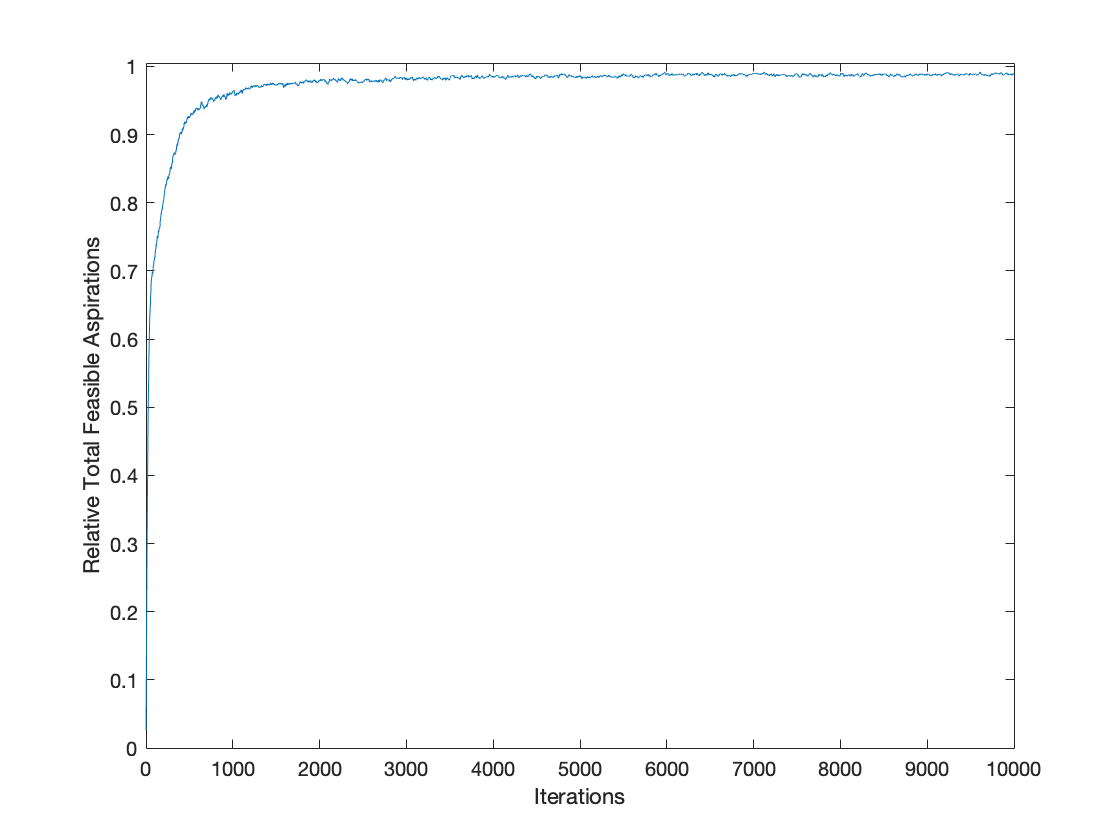}
    \caption{The average of the relative total feasible aspirations achieved by the $B$-Matching Proposals dynamics across $100$ runs.}
    \label{fig:enter-label}
    \end{center}
    
\end{figure}

\subsection{Parameters effects}
We ran supplementary simulations for general $B$-matching problems to help us gain insights into the effects of some of the problem parameters on the average performance of the dynamics, specifically the $B$-values, number of nodes, and discretization size. Figures 3, 4, and 5 illustrate the observed trends resulting from the parameters variations. 

First, to isolate the effect of varying the $B$-values, we generated $50$ random $B$-matching instances; we set the number of nodes per each class to be $9,$ yet varied the range of $B$-values. For each instance, we ran the dynamics for $3$ different $B$-value functions, $B_1,B_2,$ and $B_3$. $B_1$ value is randomly generated such that $B_1(p)\in \{1,2,3\}\ \forall p \in G,\ B_2$ is such that $B_2(p)=2B_1(p)\ \forall p \in G,$ and $B_3$ is such that $B_3(p)=3B_1(p)\  \forall p \in G$. Fig. 3 shows the average of the relative welfare across the $50$ runs while using the different $B$-value functions. Empirically, it appears for this setup that, on average, bigger $B$-values lead to better performance of the dynamics.

\begin{figure}[h]
    \begin{center}
        \includegraphics[scale=0.15]{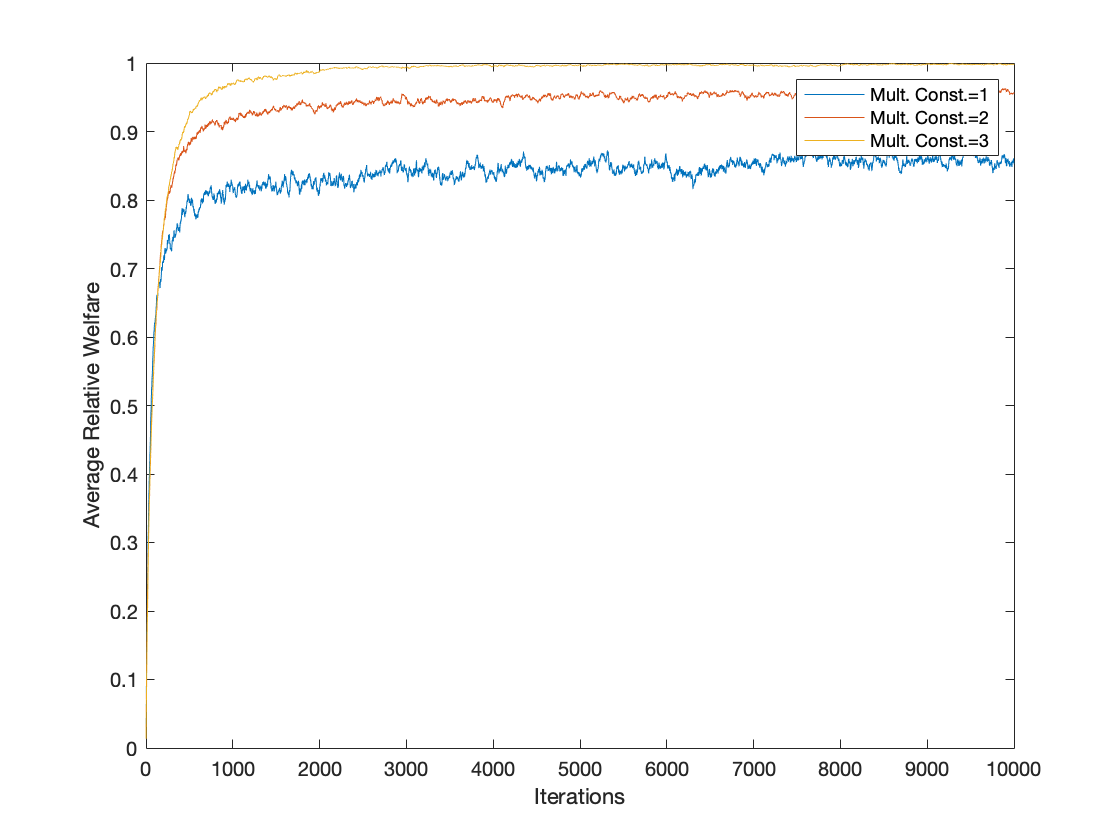}
    \caption{The average relative welfare across $50$ configurations when using $B_1, B_2,$ and $B_3$ as the $B$-value functions. The runs that use $B_1$ are represented by the blue line, those using $B_2$ are represented by the red line, and those using $B_3$ are represented by the yellow line.}
   
    \end{center}

\end{figure}

Second, we generated $50$ random $B$-matching instances, and for each instance, we ran the dynamics for $4$ different discretization values. Fig. $4$ shows the average of the relative welfare across $50$ runs while using the different discretization values. From the simulations, we observed that the smaller the discretization value is, the longer the dynamics took to converge and get closer to the welfare value.

\begin{figure}[h]
    \begin{center}

     \includegraphics[scale=0.15]{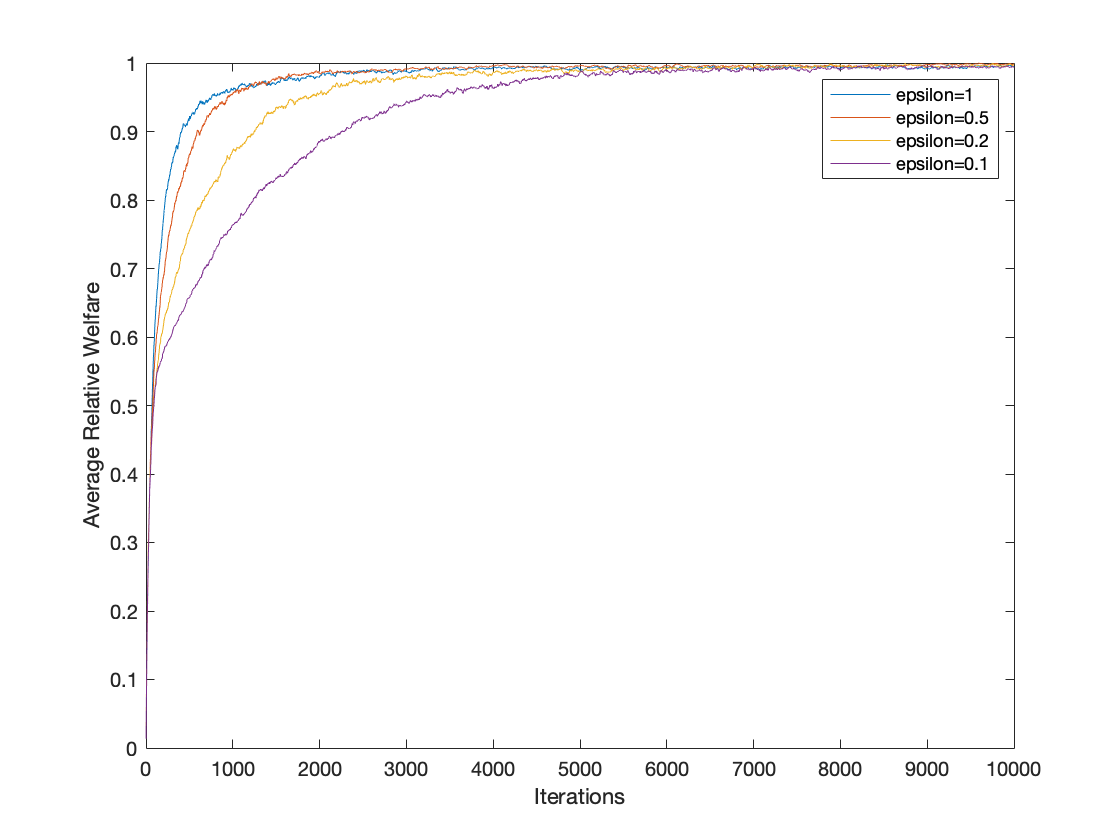}
    \caption{The average relative welfare across $50$ configurations when using different values for the discretization width. } 

    \end{center}

\end{figure}

Finally, we varied the number of nodes for $50$ random $B$-matching instances to observe the impact of decreasing the number of nodes on the dynamics performance. For each instance, we removed a number of nodes gradually and their associated edges, reran the dynamics, and compared the empirical performance. Fig. 5 illustrates the different performance levels for the dynamics resulting from varying the number of nodes.

\begin{figure}[h]
    \begin{center}
  
    \includegraphics[scale=0.15]{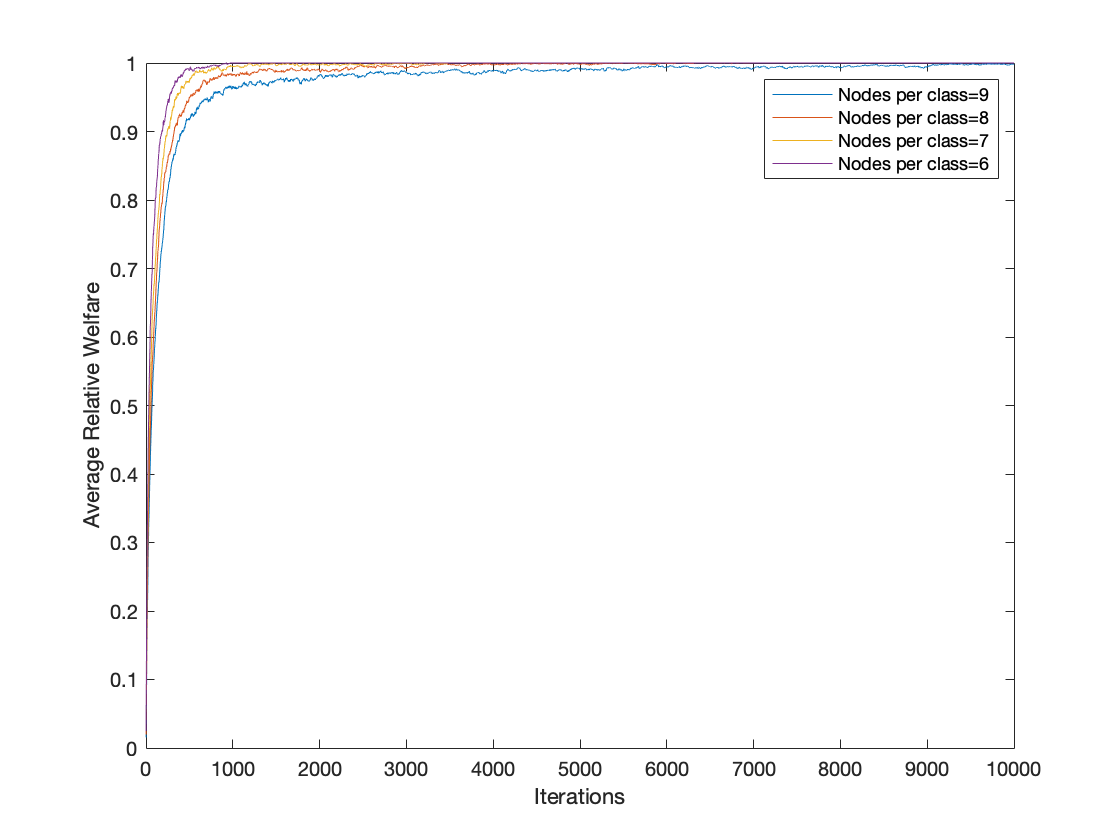}
    \caption{The average relative welfare across 50 configurations, and the dynamics performance as we gradually remove nodes and their associated edges from the graph.}
   
    \end{center}

\end{figure}

\section{CONCLUSION} To recap, we considered the $B$-matching setup. We defined a copies-core solution, which is a solution concept that specifies the gain of each node in a bipartite graph, from each match it joins. The copies-core solution achieves core stability such that no group of agents can deviate and collectively gain a higher payoff through another matching. We presented a centralized algorithm, in the nature of the Hungarian method, that reaches the copies-core solution in polynomial time. We further introduced distributed learning dynamics that converge to the core with probability one. Finally, we illustrated the distributed dynamics on a multi-agent task assignment setting.

\addtolength{\textheight}{-12cm}   





\end{document}